\documentclass[12pt]{amsart}
\usepackage{amsmath,amscd,amsthm,amsfonts, amssymb,amsxtra, mathrsfs, amsrefs}
\usepackage{hyperref}
 \hypersetup{colorlinks=true,citecolor=blue,linkcolor=olive}
\usepackage{accents}
\usepackage{caption}
\usepackage{tikz}
\usetikzlibrary{patterns,calc,decorations.markings, graphs, arrows,shapes}
\usepackage[us,12hr]{datetime}
\usepackage[all]{xy} \SelectTips{cm}{}

   \usepackage{graphicx}
  \usepackage[margin=1.75in]{geometry}
  \usepackage[pagewise, displaymath, mathlines]{lineno}
\usepackage{nicematrix}
   \usepackage[shortlabels]{enumitem}
  
  \makeatletter
\renewcommand{\fnum@figure}{Fig.~\thefigure}
\makeatother
\CDat
\newtheorem{thm}{Theorem}[section]  
\newtheorem*{un-no-thm}{Theorem}
\newtheorem{cor}[thm]{Corollary}     % Numbered along with thm
\newtheorem{lem}[thm]{Lemma}         % Numbered along with thm
\newtheorem{prop}[thm]{Proposition}

\newtheorem{bigthm}{Theorem}
\theoremstyle{definition} 
\newtheorem{defn}[thm]{Definition}   % Numbered along with thm

\theoremstyle{definition}
\theoremstyle{definition}
\theoremstyle{remark}
\newtheorem{rem}[thm]{Remark}

\newtheorem{notation}[thm]{Notation}
\newtheorem*{acks}{Acknowledgements}

\newtheorem*{out}{Outline}

\newtheorem*{intro-rem}{Remark}
\newtheorem*{intro-rems}{Remarks}

\newtheorem{ex}[thm]{Example}

\DeclareMathOperator{\spec}{Spec}

\DeclareMathOperator{\Span}{span}

\DeclareMathOperator{\self}{End}

\DeclareMathOperator{\U}{U}
\DeclareMathOperator{\Or}{O}
\DeclareMathOperator{\SO}{SO}

\DeclareMathOperator{\GL}{GL}

\DeclareMathOperator{\PGL}{PGL}

\DeclareMathOperator{\PU}{PU}

\DeclareMathOperator{\tr}{tr}

\counterwithin{figure}{section}

\def\:{\colon\!}

\def\Bbb{\mathbb}
\def\scr{\mathscr}

\def\Lcirc{\accentset{\circ}{\scr L}}

\begin{document}

\title[Rational invariants of quantum systems of $n$-qubits]{On the rational invariants of quantum systems
of $n$-qubits}
\date{\today} 
%\date{\today, \currenttime} 
\author[L.~Candelori, V.~Y.~Chernyak, and J.~R.~Klein]{
Luca Candelori$^{1}$,~Vladimir Y. Chernyak$^{1,2}$,~John R. Klein$^{1}$
}

\thanks{\hskip -.17in $^{1}$\textsc{\tiny Department of Mathematics, Wayne State University, Detroit, MI 48201, USA}\\
$^{2}$\textsc{\tiny Department of Chemistry, Wayne State University, Detroit, MI 48201, USA}}


\begin{abstract}  For an $n$-qubit system, a rational function on the space of mixed states 
which is invariant with respect to the action of the group of  local symmetries
may be viewed as a detailed measure of entanglement.
We give two proofs that the field of  all  such invariant rational functions
is purely transcendental over the complex numbers and has transcendence degree $4^n - 3n-1$. An explicit transcendence basis is  also exhibited.
 \end{abstract}

\maketitle
\setcounter{tocdepth}{1}
{\color{olive}
\tableofcontents}
\addcontentsline{file}{sec_unit}{entry}

\section{Introduction} \label{sec:intro}
A physical system is often described in terms of three notions: states, observables and dynamics. 
 We shall consider here the case of quantum systems consisting of $n$  parties,
in which each party is in possession of a  {\it qubit}, i.e., a two dimensional complex vector space.
The qubits are not permitted to interact.
Stated differently, each qubit is allowed to evolve independently of the other qubits by means of its owner.
Mathematically, this means that the time evolution of the system is a one-parameter family of {\it local} unitary transformations.

From the point-of-view of this paper, {\it entanglement} occurs when 
one considers the problem of what the parties 
can or cannot do to transform one state into another one. A solution to this problem amounts to  describing
the orbit structure of the state space with respect to action of the group of local unitary transformations, the latter
which is  isomorphic to a product of $n$-copies of the Lie group $\U(2)$. 
Consequently, any reasonable measure of entanglement should be a function on the space of orbits.
In particular, the most detailed measure of entanglement is given by considering functions
on the state space which are invariant with respect to the group of local unitary transformations.

In addition, one observes that the action of the local unitary group is trivial on the subgroup given by $n$-copies
of the group of uni-modular transformations. Hence, the orbit structure is defined by the action
of an $n$-fold product of projective unitary groups $\PU(2) \cong \SO_3(\Bbb R)$.
However, from the viewpoint of algebraic geometry, it is usually simpler to work over the algebraically closed field of complex numbers.
In this instance, one needs to complexify the space of states as well as the symmetry group, resulting in
an action of an $n$-fold product of copies of $\SO_3(\Bbb C) \cong \PGL_2(\Bbb C)$ acting the state space of 
trace one endomorphisms of the  tensor product of the $n$-qubits. If one wishes to recover
measures of entanglement over the real numbers, it is
necessary to keep track of the real structure on the complexification (for an alternative approach, see \cite{CVKR-rationality}).

\subsection{The rationality problem} If one restricts the class of invariant functions to be polynomials, or more generally,
rational functions, then the problem of determining such measures of entanglement
may be attacked using  the machinery of  geometric invariant theory.

Suppose $G$ is a reductive algebraic group and $W$ is a finite dimensional $G$ representation. Then a fundamental 
question  is whether or not the field $\Bbb C(W)^G$, i.e.,
the $G$-invariant rational functions on $W$, is a purely transcendental extension of $\Bbb C$. 
If it is, then we say $\Bbb C(W)^G$ is {\it purely rational}.
If $G$ is connected, then
a surprising fact  is that there is no known example of a $G$-module $W$ for which  $\Bbb C(W)^G$ not purely rational
\cite[p.~6]{böhning2009rationality}.

In this paper we solve the rationality
problem in the affirmative in the case when $W$ is space of (mixed) states associated with a quantum system on $n$-qubits.

\subsection{The main result}
Let $V_1,\dots,V_n$ denote $n$-qubits and let 
\[
\scr L_n := \scr L_n(V_1,\dots,V_n)
\] denote
the  affine space consisting of all trace one endomorphisms of the tensor product
$V_1 \otimes \cdots \otimes V_n$.  This is the space of {\it L-states} (or {\it mixed states}) on $n$-qubits.

For a vector space $V$, let $\PGL(V)$ denote the 
associated projective linear group.
Then the Lie group
\[
G = \PGL(V_1) \times \cdots \times \PGL(V_n)
\] 
acts on $\scr L_n$ by conjugating endomorphisms. 

The following, which is our main result, settles the rationality problem for $G$ acting on $\scr L_n$.
  
  \begin{bigthm} \label{bigthm:main} The field $\Bbb C(\scr L_n)^G$ 
  is purely rational and has transcendence
  degree $4^n - 3n -1$ over the complex numbers.
  \end{bigthm}
  
  \begin{rem} To the best of our knowledge, Theorem \ref{bigthm:main}  does not directly follow from known results on the rationality problem \cite{böhning2009rationality}, \cite{CT-S}.  
  
  If $n=2$, then Theorem \ref{bigthm:main} appears in \cite[thm.~VIII.1]{CCKN-effective}. 
The proof given there relies on two ingredients: (i) the existence of a rational section (or slice) for the action on which a discrete 
semi-simple finite abelian Weyl group acts, and (ii) the ``No-Name Lemma'' \cite{dolgachev}, the latter 
  which is a tool used to identify the field of rational invariants of the total space of a $G$-linearized vector bundle whose base space is a $G$-torsor.
The proof proceeds by means of determining the algebraic quotient of the  action of the Weyl group on the rational section.

In this paper, we provide two proofs of Theorem \ref{bigthm:main}. Our first
proof dispenses with a choice of rational section altogether and is more direct. 
An added feature of the first proof is that it enables one to construct an explicit transcendence basis for  $\Bbb C(\scr L_n)^G$. This is accomplished in  \S\ref{sec:basis}.  Our second proof
is similar in spirit to the proof of \cite[thm.~VIII.1]{CCKN-effective} in that it uses both the No-Name Lemma and a judicious
choice of rational section.
  \end{rem}

   \begin{out} In \S\ref{sec:prelim} we provide an alternative description of the space $\scr L_n$, called the Bloch model, 
   which expresses a mixed state in terms of its set of correlation functions. In \S\ref{sec:quotients} we  briefly discuss
   geometric invariant theory in the affine case \cite[ch.~\!\!1\S2]{GIT}. Then \S\ref{sec:rationality} contains a proof of Theorem \ref{bigthm:main}.
In \S\ref{sec:basis}, we provide a transcendence basis for $\Bbb C(\scr L_n)^G$. Lastly, \S\ref{sec:second} sketches
our second proof of Theorem \ref{bigthm:main}.
   \end{out}

\begin{acks}
  The authors are supported by the U.S. Department of Energy, Office of Science, Basic Energy Sciences, under Award Number DE-SC-SC0022134.
\end{acks}
  
\section{States and observables}  \label{sec:prelim}
The purpose of this section is to describe the state space $\scr L_n$, the associated
observables, as well as the kinds of correlation functions
that arise in considering quantum systems of $n$-qubits. We also
explain how the  complexified local unitary group acts on $\scr L_n$.

 A {\it qubit} is a complex vector space of dimension two. If $V$ is a qubit, then
we write 
\[
\scr V := \mathfrak{sl}(V)
\] for the Lie algebra of {\it  traceless endomorphisms} of $V$.
The {\it analytical scalar product} $\scr V\otimes \scr V \to \Bbb C$ is given by the Killing form
\[
\langle A,B\rangle := \tfrac{1}{2}\tr(AB)\, ,
\]
in which $\tr$ denotes the trace of an operator.
When $A = B$, we will resort to the notation $\lVert A\rVert^2 := \langle A, A\rangle$.

The {\it vector product} 
\[
[{-},{-}]\: \scr V\otimes \scr V \to \scr V
\]
is  the Lie bracket $[A,B] = AB-BA$.  Then the triple scalar product
\begin{align*}
\Lambda^3 \scr V &@> \cong >> \Bbb C\\
A\wedge B \wedge C &\mapsto \langle [A,B],C\rangle
\end{align*}
equips $\scr V$ with a canonical orientation.

\begin{rem} A choice of ordered basis for $V$ induces an isomorphism of Lie algebras
 $\mathfrak{sl}_2(\Bbb C) \cong \scr V$. Then  $\scr V$ is 
 equipped with a positively oriented ordered basis of  {\it Pauli matrices}
\[
\sigma_x = \begin{pmatrix}
0& 1\\
1 & 0
\end{pmatrix},
\sigma_y = \begin{pmatrix}
0 & -i\\
i & 0
\end{pmatrix} \! ,
\sigma_z = \begin{pmatrix}
1 & 0\\
0 & -1
\end{pmatrix}\! .
\]
\end{rem}

Let $\SO(\scr V)$ be  the special orthogonal group of $\scr V$, i.e., the 
group of orientation preserving isometries $\scr V \to \scr V$. 

\begin{rem} 
Let $\PGL(V) = \GL(V)/\Bbb C^\times$ denote the projective general linear group of 
the qubit $V$.
Then the action $\PGL(V)  \times \scr V \to \scr V$ defined by
conjugating traceless endomorphisms is orientation preserving and preserves the
analytical scalar product. Hence, the adjoint of this action 
defines a canonical isomorphism of Lie groups
\[
\PGL(V) \cong \SO(\scr V)\, .
\]
\end{rem}

\subsection{L-states}   For a vector space
$V$, we let $\self(V)$ denote the vector space of endomorphisms $V\to V$.
Let $\overline{n} = \{1,\dots n\}$. Fix  $n$ qubits
$V_1,\dots V_n$. Then associated with each
$V_j$ we have $\scr V_j$ as above.

For $I \subset \overline{n}$ we set
\[
V_I = \bigotimes_{i\in I} V_i \, ,\qquad \scr V_I = \bigotimes_{i\in I} \scr V_i\, .
\]
We define the space of {\it L-states on $n$-qubits}
\[
\scr L_n := \scr L(V_1,\dots,V_n) \subset \self(V_{\overline{n}})
  \]
as the codimension one affine subspace consisting of the trace one endomorphisms
of $V_{\overline{n}} = V_1 \otimes \cdots \otimes V_n$. We refer to vectors
$\rho \in \scr L_n$ as {\it L-states}; these are also known as the {\it mixed states} of the quantum system. 
Note that as an algebraic variety, $\scr L_n$  
is isomorphic to the affine space $\Bbb A^{4^n-1}$.

We set $G_j =  \PGL(V_j) \cong \SO(\scr V_j)$ and define
\[
G := \prod_{j=1}^n G_j \cong  \prod_{j=1}^n  \SO(\scr V_j) .
\]
Then $G$ acts on $\scr L_n$ by conjugating endomorphisms. Moreover, $G$ is
the complexification of the group of local unitary transformations.

\subsection{Observables}
The space $\Omega$ of {\it (global) observables} 
 is the vector space   $\self(V_{\overline{n}})$.
 Given an observable $A$, the expected value of measuring $A$ in an L-state $\rho$ is
given by 
\[
\tr(A\rho) \in \Bbb C\, .
\]
A {\it local observable} at $j\in \overline{n}$ is an element of $\scr V_j$.   A local observable $A\in \scr V_j$
determines a global observable by forming the tensor product $A\otimes \text{id}$,
where  $\text{id}$  is the identity operator of $V_{\overline{n} \setminus \{i\}}$.

More generally, 
let
\[
\binom{\overline{n}}{k}
\]
denote the space of subsets of of $\overline{n}$ of cardinality $k$.  Then one has a  vector bundle
\[
E_k \to \tbinom{\overline{n}}{k}
\]
whose fiber at $I\in  \binom{\overline{n}}{k}$ is
$
\scr V_I
$.
A {\it $k$-point observable} at the configuration $I$ is an element of $\scr V_I$. A $k$-point observable at $I$
determines a global observable by tensoring it with the identity operator of $\scr V_{\overline{n} \setminus I}$. 

\begin{rem} 
It is customary to identify a $k$-point observable with its associated global observable.
\end{rem}

\subsection{Correlation functions} 
For a given $\rho \in \scr L_n$ and a choice of section $s\: \binom{\overline{n}}{k} \to E_k$,
we obtain a function
$
 \hat s\: \binom{\overline{n}}{k} \to \Omega
$
which assigns to $I$ the global observable associated with $s(I)$.
A section $s$ is said to be {\it supported} at $I$ if it vanishes
whenever $J \ne I$. In this case,  specifying $s$ is equivalent to specifying 
a $k$-point observable at $I$. 
Moreover, to any section $s$ we may associate a section supported at $I$, namely,
the one associated with the observable $s(I)$.
 
The  {\it $k$-point correlation function} of $\rho$ associated with a section $s$  is the map
\begin{align*}
\tbinom{\overline{n}}{k}  & @>a(\rho,s)>>\Bbb C\, ,\\
                     I                          &\mapsto \tr(\hat s(I)\rho) 
\end{align*}
which describes the expected values of  a family of $k$-point observables.
We note that  $a(\rho,s)(I)$ only depends on $\rho$ and the $k$-point observable $s(I)$.

If we fix $\rho$ and $I$, then
the assignment $s\mapsto a(\rho,s)(I)$ defines a linear functional $\scr V_I \to \Bbb C$.
Consequently, we have defined a function
\begin{align} \label{eqn:correlate}
\begin{split}
\scr L_n &\to \scr V_I^* \\
\rho & \mapsto (s\mapsto a(\rho,s))\, ,
\end{split}
\end{align}
which assigns to an $L$-state the totality of its $k$-point correlation functions evaluated at $I$.

The vector space of sections of $E_k \to \binom{\overline{n}}{k}$
supported at $I$, i.e., the space $\scr V_I$,
admits a basis consisting of $3^k$  vectors:
Choose an ordered basis for each $V_i$. Then each $\scr V_i$ inherits a basis of Pauli matrices 
$\{\sigma^i_x,\sigma^i_y,\sigma^i_z\}$ and $\scr V_I$ has an associated basis given by 
$3^k$ basic tensors of the form
\[
\sigma_{I,\phi} :=  \underset{i\in I}\otimes \sigma^i_{\phi(i)}\, ,
\]
where $\phi \: I \to \{x,y,z\}$ ranges over all functions. 

If we let $I$ and $k$ vary, then it is straightforward to check that  $\rho$ is completely determined by
the collection of  $4^n-1$ correlations $a(\rho,s)$, where $s$ varies 
through the sections supported on any non-empty finite subset of ${\overline{n}}$.

\begin{ex}[$n=2$] In this instance, for a given $\rho$ it is convenient to display the correlation functions
as follows: Set $\sigma_0 = I, \sigma_1 = \sigma_x, \sigma_2 = \sigma_y, \sigma_3 = z$, and
set $\rho_{ij} = \tr((\sigma_i \otimes \sigma_j) \rho)$. Then we obtain a $4\times 4$ complex matrix
of correlations
\[
\begin{pNiceArray}{c|ccc}
1 & \rho_{01} & \rho_{02} & \rho_{03} \\
\hline
\rho_{10} & \rho_{11} & \rho_{12} & \rho_{13} \\
\rho_{20} & \rho_{21} & \rho_{22} & \rho_{23}\\
\rho_{30} & \rho_{31} & \rho_{32} & \rho_{33}  
\end{pNiceArray}
\]
in which $\rho_{00} = \tr(\rho) = 1$. In the above, the displayed $1\times 3$  and $3\times 1$ 
blocks list the one-point correlations, and the $3\times 3$ block 
lists the two-point correlations.
\end{ex}

\subsection{The Bloch model} There is an alternative way of describing L-states
that is equivalent to the correlation function description.
The idea is to implement the identifications $\self(V_i) \cong \Bbb C \oplus \scr V_i$. These
 induce a canonical decomposition 
\[
\Omega \,\,  \cong\,\,   \bigoplus_{I\subset {\mathbf n}} \scr V_I \, .
\]
Moreover, the constraint that an endomorphism be of trace one on the left
corresponds to excising the summand $\scr V_\emptyset$ on the right. 

\begin{defn} The {\it Bloch model} for L-states on $n$-qubits is
\[
\scr L_n^b := \bigoplus_{\emptyset \ne I\subset {\mathbf n}} \scr V_I \, .
\]
\end{defn}

\begin{ex}[$n = 2$] The space $\scr L^b_2$ is given by triples
\[
(a_1,a_2,C)
\]
in which $a_i \in \scr V_i$ and $C\in \scr V_1 \otimes \scr V_2$.
\end{ex}

In summary, there is a canonical isomorphism
\begin{align*}
\scr L_n &@> \cong >> \scr L_n^b\, ,\\
\rho       &\mapsto \rho^b\, .
\end{align*}

\begin{rem} The action of $G$ on $\scr L_n^b$ preserves the summands $\scr V_I$.
To determine how $G$ acts on $\scr V_I$, 
recall the canonical isomorphism
\[
 G \cong \prod_{j=1}^n \SO(\scr V_j)\, .
\]
For $I\subset \overline{n}$, let $G_I = \prod_{j\in I} \SO(\scr V_j)$. Then the action of $G$ on 
$\scr V_I$ is by $G_I$ using the projection $G\to G_I$.
\end{rem}

\begin{rem}
Henceforth, we only consider the model for L-states given by $\scr L_n^b$ and we will drop the superscript from the notation.  
\end{rem}

\begin{defn} For $I \subset \overline{n}$, the canonical projection
\begin{align*}
\scr L_n &\to \scr V_I\\
\rho &\mapsto \rho_I
\end{align*}
defines the {\it $I$-component} of an L-state.
\end{defn}

\begin{rem} With respect to the canonical isomorphism $\scr V_I^\ast \cong \scr V_I$,
the map $\scr L_n \to \scr V_I^\ast$  of \eqref{eqn:correlate} coincides with the projection onto the  $I$-component. 
\end{rem}

\section{Quotients and rationality} \label{sec:quotients}
Here we assemble some necessary results from geometric invariant theory 
\cite[ch.~\!\!1\S2]{GIT}, \cite{PV} in the affine case. 
In what follows, $k$ will be
an algebraically closed field of characteristic zero. A $k$-{\it variety} $X$ is a scheme over $k$
 which is integral (i.e., irreducible and reduced) and of finite type.

\subsection{Invariant functions and quotients}
\label{sec:sections}
Let $X$ be a $k$-variety and let $G$ an affine algebraic group over $k$ acting (algebraically) on $X$.
Recall that $G$ is said to be {\it reductive} if it possesses a faithful semi-simple representation over $k$.
In what follows, we will always assume that $G$ is reductive.

  We write the coordinate ring (i.e., the global sections of the structure sheaf) of $X$ by $k[X]$  
  and the associated function field by $k(X)$.
  We denote the ring of polynomial $G$-invariant functions on $X$ by $k[X]^G$, so $k[X]^{G}$ is a $k$-subalgebra of $k[X]$.
  Similarly, we write the field of rational $G$-invariant functions by $k(X)^{G}$.
  
 \begin{defn} The {\it algebraic quotient} of $X$ by $G$ is the affine scheme
 \[
X/\!\!/G :=  \spec(k(X)^G)\, .
  \]
  We define a {\it rational quotient} for the action of $G$ on $X$ to be
  a $k$-variety$Y$ which is birationally equivalent to the
  affine quotient. Equivalently, $k(Y) \cong k(X)^G$.
  \end{defn}  
  A fundamental question regarding the field of rational invariants $k(X)^G$ is whether it is \emph{purely rational}, i.e.,
   whether $k(X)^{G}$ is generated by a finite number of algebraically independent rational functions. This is equivalent
   to the question of whether or not there is a rational quotient for the action of $G$ on $X$
   that is isomorphic to an affine space $\Bbb A^m$.
If this is the case, then the transcendence degree
  of the field $k(X)^G$ over $k$ is $m$.

\subsection{Rationality }
 A useful tool for proving rationality is the following version of descent for vector bundles along a torsor.

\begin{prop}[Descent {\cite{dolgachev}, \cite[2.13]{PV},  \cite{CT-S}}]\label{prop:no-name}
  Let $G$ be a reductive group that acts almost freely on a $k$-variety $X$.
    Let $\pi \: E\to X$ be a $G$-linearized vector bundle. Then the algebraic quotient
    $E/\!\!/G$ is birationally equivalent
    to  a vector bundle over the algebraic quotient $X/\!\!/G$.   
    \end{prop}
    
    \begin{rem} Let $d$ be the rank of  $\pi$.
    Then with respect to the assumptions,
    another formulation of Proposition \ref{prop:no-name} is that there is a commutative diagram of $k$-varieties with $G$-action
    \[
    \xymatrix{
    U\times \Bbb A^d \ar[dr]_{p_1} \ar[rr]^\cong &&  \pi^{\ast}U\ar[dl]^\pi \\
   &  U
   }
   \]
   in which $U \subset X$ is an open subset on which $G$ acts freely,  $p_1$ is first factor projection,
   $G$ acts trivially on $\Bbb A^d$, and the horizontal map is an isomorphism of vector bundles.
   
Note that the construction of such a diagram is equivalent to finding a basis of equivariant sections of 
$\pi$ along the open set $U$.
\end{rem}

    \begin{cor}[``No-Name Lemma'' {cf.~\cite[cor.~3.8]{CT-S}}] \label{cor:no-name} 
There is an isomorphism of function fields
    \[
 k(E)^G \cong k(\Bbb A^d) \otimes_k k(X)^G\, .
    \]
    \end{cor}

\section{First proof of Theorem \ref{bigthm:main}}  \label{sec:rationality}

A {\it directed graph} $\Gamma $ consists of a pair of sets $(\Gamma_1,\Gamma_0)$ together
with a function 
\[
(d_0,d_1)\:  \Gamma_1\to \Gamma_0 \times \Gamma_0
\]
which is  one-to-one. 
Here $\Gamma_1$ is the set of (directed) edges and $\Gamma_0$ is the set of vertices.
The function $d_0$ assigns to an edge its incoming vertex, and $d_1$ assigns to an edge
its outgoing vertex.

Let $K_n$ denote the complete directed graph whose set of vertices is $\overline{n}$.
Then $K_n$ has  $n^2-n$ edges labeled by ordered pairs $(i,j)$, with $i\ne j$.
For such an edge, $i$ is the incoming vertex and $j$ is the outgoing vertex. If there is no confusion, we frequently
 abbreviate notation and denote the edge $(i,j)$  by $ij$. Then $d_0(ij) = i, d_1(ij) =j$.
 
\begin{defn} A {\it nowhere zero vector field} on $n$ vertices is
 a directed subgraph  $\Gamma \subset K_n$ such that
\begin{itemize}
\item the set of vertices of $\Gamma$ is ${\overline{n}}$, and
\item to each vertex $i$ of $\Gamma$, there is 
precisely one edge of the form $ij$, i.e., each vertex has precisely one outgoing edge.
\end{itemize}
In particular, there is a bijection between the set of vertices and the set of edges of $\Gamma$.
\end{defn} 
The second condition in the definition is equivalent to requiring
 $d_0 \: \Gamma_1 \to \Gamma_0$ to be an isomorphism or equivalently,
requiring the Euler characteristic of $\Gamma$ to be trivial.
 
 \begin{rem}
A nowhere zero vector field is given by
the cyclic graph indicated by sequence of arrows
\[
1 \to 2 \to \cdots \to n \to 1\, .
\]
In fact, nowhere zero vector fields are abundant:
There is a
a one-to-one correspondence between nowhere zero vector fields and
the set of right inverses of function $d_0\: (K_n)_1 \to (K_n)_0$. If $s$ is a right inverse, i.e. $d_0s = 1$,
then the corresponding nowhere zero vector field consists of all edges in the image of $s$.
In particular, there are precisely $(n-1)^n$ nowhere zero vector fields
on $n$ vertices.
\end{rem}

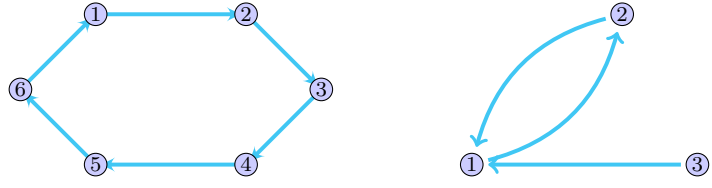
\begin{figure}
\begin{tikzpicture}
\tikzset{vertex/.style = {shape=circle,draw,minimum size=.5em}}
\tikzset{edge/.style = {->,> = latex'}}

\draw [edge, cyan!60, line width=1.5] (0,0) to (2,0);
\draw [edge, cyan!60, line width=1.5]  (2,0) to (3,-1) ;
\draw [edge, cyan!60, line width=1.5]  (3,-1) to (2,-2);
\draw [edge, cyan!60, line width=1.5]   (2,-2) to (0,-2);
\draw [edge, cyan!60, line width=1.5]  (0,-2) to (-1,-1); 
\draw [edge, cyan!60, line width=1.5]  (-1,-1) to (0,0); 

\draw[fill=blue!20!white] (0,0) circle [radius=.15];
  \draw[fill=blue!20!white] (2,0) circle [radius=.15]; 
  \draw[fill=blue!20!white] (3,-1) circle [radius=.15]; 
  \draw[fill=blue!20!white] (2,-2) circle [radius=.15]; 
  \draw[fill=blue!20!white] (0,-2) circle [radius=.15]; 
  \draw[fill=blue!20!white] (-1,-1) circle [radius=.15]; 

 \node at (0,0) {\tiny 1};
    \node at (2,0) {\tiny 2}; 
     \node at (3,-1) {\tiny 3};  
     \node at (2,-2) {\tiny 4}; 
      \node at (0,-2) {\tiny 5}; 
      \node at (-1,-1) {\tiny 6}; 

\draw[fill=blue!20!white] (5,-2) circle [radius=.15];
  \draw[fill=blue!20!white] (7,0) circle [radius=.15]; 
  \draw[fill=blue!20!white] (8,-2) circle [radius=.15]; 

 \node (a) at (5,-2) {\tiny 1};
    \node (b) at (7,0) {\tiny 2}; 
     \node (c)  at (8,-2) {\tiny 3};  

\draw[->, cyan!60, line width=1.5]
(a) edge[bend right] (b)
(c) edge (a) 
(b) edge[bend right]  (a);
\end{tikzpicture}

\caption{Nowhere zero vector fields on 6 vertices and on 3 vertices.}  \label{fig:one}
\end{figure}

In what follows we fix $\Gamma$, a nowhere zero vector field on $n$ vertices.

Let $\breve{\scr V}_i \subset \scr V_i$ denote the set of vectors having non-zero squared norm, and let
$\Lcirc \subset \scr L$ denote the Zariski open subset 
given by those L-states $\rho$
with the property that each one-point component $\rho_j \in \scr V_j$ has non-zero squared norm.

Label the vertices of $\Gamma$ by vectors $a_k \in \breve{\scr V}_k$, for $k = 1,\dots,n$. 
Then $a_k$ defines an orthogonal direct sum decomposition
\[
\scr V_k = \scr V_k^\ell \oplus \scr V_k^t\, ,
\]
in which $ \scr V_k^\ell =\Span(a_k)$.  With respect to the choice
of $a_k$, we call $\scr V_k^\ell$ the {\it longitudinal} of $\scr V_k$ and 
$\scr V_k^t$ the {\it transversal} of $\scr V_k$.

For each choice of such vertex labels,  we label an edge $ij$  by
the two dimensional vector space
\[
\scr V^{\ell t}_{ij} := \scr V_i^\ell \otimes \scr V_j^t   \, .                 
\]
In particular, $\scr V^{\ell t}_{ij}$   is canonically identified with $\hom(\scr V_i^\ell,\scr V_j^t)$ using the identification
$\scr V_i^\ell \cong (\scr V_i^\ell)^*$ arising from the analytical scalar product.

Let $\scr B$ be the space of such labelings. Then
\[
\scr B \subset \prod_{i\in \Gamma_0} \breve{\scr V}_k\,\,  \times \,\,   \prod_{ij\in \Gamma_1} \scr V_i \otimes \scr V_j\, 
\]
and projection onto the first factor defines a vector bundle  $\scr B \to \prod_i \breve{\scr V}_i$.
The fiber of this vector bundle at $a_\bullet := (a_1,\dots, a_n) \in \prod_i \breve{\scr V}_i$ is the direct sum 
\[
\bigoplus_{ij\in \Gamma_1} \scr V^{\ell t}_{ij}\, ,
\]
where $\scr V^{\ell t}_{ij}$ depends on $a_i$ and $a_j$.

The evident projection
\[
p\: \Lcirc_n \to \scr B
\]
 is a $G$-linearized vector bundle. The fiber of $p$ is identified with the vector space
\[
F:= \bigoplus_{ij\in \Gamma_1, i < j} ((\scr V_i^t \otimes \scr V_j^t) \oplus (\scr V_i^\ell \otimes \scr V_j^\ell)) \oplus \bigoplus_{|I|> 2} \scr V_I
\]
which has dimension $4^{n} - 5n-1$.

\begin{lem} \label{lem:no-name} There is an isomorphism 
\[
\Bbb C(\scr L_n)^G \cong \Bbb C(\scr F)  \otimes \Bbb C(\scr B)^G \, .
\]
\end{lem}

\begin{proof} 
The group $G$ acts freely on $\scr B$, so one obtains a $G$-torsor $\scr B \to \scr B/G$.
The result now follows by the ``No-Name Lemma'' \cite[cor.~3.8]{CT-S}
\end{proof}

Let
\[
\mathfrak F_j  = \{(v,w)|v,w \in \breve {\scr V}_j, \langle v,w\rangle = 0\}\, .
\]
Then $\mathfrak F_j$ is a variety.

Suppose $ij$ is an edge of $\Gamma$. 
Given a triple $(a_i,a_j,C_{ij}) \in \breve{\scr V}_i \times \breve{\scr V}_j \times {\scr V}^{\ell t}_{ij}$
we may associate the vector
\[
b_j := [C_{ij}\lrcorner a_i,a_j] \in {\scr V}^t_j\, .
\]
Here $C_{ij}\lrcorner a_i$ denotes contraction, and $[{-},{-}]$ denotes the Lie bracket of
${\scr V}_j$.  
We define
\[
\accentset{\circ}{\scr B} \subset \scr B
\]
the Zariski open subset given
by those $(a_\bullet,C_\bullet )$ such that $\lVert b_j\rVert^2 \ne 0$ for all $j$. 

Then we have  a $G$-equivariant map
\begin{align} \label{eqn:iso}
\begin{split}
\accentset{\circ}{\scr B}  &\to \prod_{j=1}^n \mathfrak F_j  \\
(a_\bullet,C_\bullet) &\mapsto ((a_1,b_1),\dots, (a_n,b_n))\, ,
\end{split}
\end{align}
where $G$ acts on $j$-th factor  $\mathfrak F_j$ via the projection
$G\to \SO(\scr V_j)$.

\begin{lem}  \label{lem:scheme-iso} The map $\accentset{\circ}{\scr B}  \to \prod_{j=1}^n \mathfrak F_j $ is an isomorphism of schemes.
\end{lem}

\begin{proof} As above, we may regard a tensor $C_{ij} \in \scr V_{ij}^{\ell t}$
as a homomorphism $ C_{ij} \: \scr V_i^\ell \to \scr V_j^t$. With respect this indentification
the contraction with $a_i$ is given by the evaluation $C_{ij}(a_i)$.

Let $(a_j,b_j) \in \mathfrak F_j$ be given. Then a homomorphism
$\scr V_j^\ell \to \scr V_j^t$ is completely determined by its value at $a_j\in \scr V_j^\ell  $.
Moreover, the map
\begin{align*} 
\scr V_j^t & \to \hom(\scr V_j^\ell ,\scr V_j^t) \\
v & \mapsto (a_j \mapsto [v,a_j])
\end{align*}
is an isomorphism of vector spaces. In particular, the homomorphism determined by $a_j \mapsto b_j$ corresponds
under this isomorphism to an element $c\in \scr V_j^t$ which has the property that $[c,a_j] = b_j$. Setting $C_{ij}(a_j) := c$,
we obtain a unique homomorphism $C_{ij} \: \scr V_i^\ell \to \scr V_j^t$ that satisfies $[C_{ij}(a_i),a_j] = b_j$.

Next, given $((a_1,b_1),\dots, (a_n,b_n)) \in \prod_{j=1}^n \mathfrak F_j $, 
we obtain $C_{ij}$ for all edges of $\Gamma$ satisfying $[C_{ij}(a_i),a_j] = b_j$
 This recipe defines the inverse homomorphism
$\prod_{j=1}^n \mathfrak F_j  \to \accentset{\circ}{\scr B}$.
\end{proof}

In view of the lemma, it suffices to identify a rational quotient for the action of $G$ on $\prod_j \mathfrak F_j$.
As $G$ acts on the factor $\mathfrak F_j$ via the projection $G\to \SO(\scr V_j)$, it will suffice to 
identify a rational quotient for $G_j := \SO(\scr V_j)$ acting on $\scr F_j$. 

Let $\Bbb G_m = \Bbb C^\times$ denote the multiplicative
group. Consider the $G$-invariant map
\begin{align*}
\mathfrak F_j & \to \Bbb G_m \times \Bbb G_m\\
(a_j,b_j) &\mapsto (\lVert a_j \rVert^2,\lVert b_j\rVert^2)\, .
\end{align*}

\begin{prop} \label{prop:torsor} The map $\mathfrak F_j \to \Bbb G_m \times \Bbb G_m$ is $G_j$-torsor. Consequently,
a rational quotient for the action of $G$ on $\accentset{\circ}{\scr B} $ is given by
\[
\Bbb G_m^{\times 2n}\, .
\] 
\end{prop}

\begin{proof} Fix a pair $(a,b) \in \mathfrak F_j$ with the property that $\lVert a\lVert^2 = 1 = \lVert b\lVert^2$.
Let $U = \Bbb G_m \times \Bbb G_m$ and let
\[
\phi\: U \to \Bbb G_m \times \Bbb G_m
\]
be the map $(r,s) \mapsto  (r^2,s^2)$. Then $\phi$ is a \'etale cover.
Define a map
\[
\tilde \phi\: G_j \times U \to  \mathfrak F_j
\]
by $(A,r,s) \mapsto (rAa,sAb)$. 
Then $\tilde \phi$ defines a trivialization of 
$\mathfrak F_j \to \Bbb G_m \times \Bbb G_m$ along $\phi$.
The proves the first part.

As to the second part, we infer that a rational quotient $\mathfrak F_j/\!\!/G_j$ is given by
$\Bbb G_m \times \Bbb G_m$.  Consequently, using Lemma \ref{lem:scheme-iso}, a rational quotient 
for  the action of $G$ on $\accentset{\circ}{\scr B} $ is given by
\[
\prod_{j=1}^n \mathfrak F_j/\!\!/G_j \cong \prod_{j=1}^n (\Bbb G_m \times \Bbb G_m) = \Bbb G_m^{\times 2n}\, . \qedhere
\] \qedhere
\end{proof}

\begin{proof}[Proof of Theorem \ref{bigthm:main}]
By Lemma \ref{lem:no-name} there is an isomorphism
\[
 \Bbb C(\scr L)^G \cong \Bbb C(\scr F)  \otimes \Bbb C(\scr B)^G ,
\]
where $ \Bbb C(\scr F)$ is purely rational with transcendence degree $4^n - 5n-1$ over $\Bbb C$. By Proposition
\ref{prop:torsor}, $\Bbb C(\scr B)^G \cong \Bbb C(\accentset{\circ}{\scr B}/\!\!/G)$ is purely rational of transcendence degree
$2n$ over $\Bbb C$. Hence the displayed tensor product has transcendence degree $4^n - 3n-1$ over $\Bbb C$.
\end{proof}

\section{An explicit transcendence basis} \label{sec:basis}

Above we described a transcendence basis of $\Bbb C(\scr B)^G$ over $\Bbb C$ consisting of
$2n$ invariants. Using the injective homomorphism
\[
\Bbb C(\scr B)^G \to \Bbb C(\scr L_n)^G
\]
we obtain
 $2n$ algebraically independent elements of $\Bbb C(\scr L_n)^G$. We will show
 how to extend this to a transcendence basis of $ \Bbb C(\scr L_n)^G$ over $\Bbb C$
 by introducing $4^n - 5n -1$ additional invariants.

Let $R = \Bbb C[\scr L_n]$ be the coordinate ring of $\scr L_n$ and
observe that vector space $\scr V_I$ is equipped with
an analytical scalar product.

For each subset $I \subset \overline{n}$,  the canonical projection
\[
a_I\: \scr L_n \to \scr V_I
\]
corresponds to an element 
\[
a_I^*\in \hom_{\Bbb C}(\scr V_I,R) \cong \scr V_I \otimes_{\Bbb C} R
\]
where in the above isomorphism we have implicitly used the analytical scalar product.

First consider the case when $|I| \ge 3$. 
Then the analytical scalar product induces a $G$-invariant pairing
\[
(\scr V_I \otimes_{\Bbb C} R) \otimes (\scr V_I \otimes_{\Bbb C} R)@> \langle {-},{-} \rangle >> R\, .
\]
Consider the functions
\[
a_{j,r}\: \scr L_n \to \scr V_j\, , \qquad j \in \overline{n}, r \in \{0,1,2\}
\]
defined by
\[
a_{j,0}(\rho) :=  \rho_j\, , \quad a_{j,1}(\rho) := [\rho_{jk}\lrcorner \rho_k,\rho_j] \, , \quad  a_{j,2}(\rho) := [\rho_j,a_{j,1}(\rho)]\, ,
\]
where $jk \in \Gamma_1$ is the unique edge that points away from  the vertex $j$.
Given a function $\phi\: I \to \{0,1,2\}$, we may form the function
\[
a_{I,\phi} := \bigotimes_{j \in I} a_{j,\phi(j)} \: \scr L_n\to  \scr V_I\, .
\]
Then $a_{I,\phi}$ corresponds to an element  $a_{I,\phi}^* \in   \scr V_I \otimes_{\Bbb C} R$.

Applying the pairing, we obtain a $G$-invariant element
\[
\theta_{I,\phi} :=  \langle a_I^\ast,a_{I,\phi}^* \rangle \in R\, .
\]
If $I$ is fixed of cardinality $r$, then  there are $3^r$ such invariants corresponding to the number of functions $\phi$.
Hence, if $I \in \binom{\overline{n}}{r}$ is allowed to vary, there are a total $\binom{n}{r}3^r$ such invariants.

If $|I| = 2$ the above construction is to be modified.  
Suppose  $I = \{i , j\}$. 
For $\rho \in \scr L_n$, to avoid clutter we write
$C := \rho_{ij}$. Set
\[
a_j := \rho_j , \quad b_j  := [C\lrcorner a_i,a_j], \quad c_j:= [a_j,b_j]
\]
Then we obtain five invariants associated with $I$:
\begin{itemize}
\item  four invariants $C_{rs}^{tt}$  indexed over $r,s\in \{1,2\}$:
\begin{align*}
C_{11}^{tt} = \langle Cb_i, b_j\rangle\, ,  \quad &C_{12}^{tt} = \langle Cb_i, c_j\rangle\, , \\
 C_{21}^{tt} = \langle Cc_i, b_j\rangle\, , \quad & C_{22}^{tt} = \langle Cc_i, c_j\rangle\, .
\end{align*}
\item one invariant $C^{\ell\ell} :=  \langle Ca_i, a_j\rangle$.
\end{itemize}
Hence, allowing $I$ to vary through subsets of cardinality two, we obtain
a total of $5\binom{n}{2}$ invariants of the above type.

Suppose now that $ij$ is a directed edge of $K_n$ which is not an edge of $\Gamma$.
In this case we specify two invariants associated with  $ij$, namely
\begin{itemize}
\item $C_{11}^{\ell t} := \langle C\lrcorner a_i,b_j \rangle$,
\item $C_{12}^{\ell t} := \langle C\lrcorner a_i,c_j \rangle$,
\end{itemize}
As $\Gamma$ has $n$ edges, there are a total of $2(2\binom{n}{2} - n) = 2n^2-4n$ such invariants.

Then the total number
 of distinct invariants is given by
\[
5\tbinom{n}{2} +  2n^2-4n +  \sum_{j=3}^n \tbinom{n}{j} 3^j = 4^n - 5n-1\, .
\]

\begin{thm} The above invariants, together with the $2n$ invariants $\lVert a_j\lVert^2, \lVert b_j\rVert^2$ for
$j \in \overline{n}$, form a transcendence basis for 
$\Bbb C(\scr L_n)^G$ over $\Bbb C$.
\end{thm}

\begin{proof} An inspection of the construction shows that
each of $4^n - 5n-1$ invariants has the form
$\theta_{I,\phi} = \langle  a_I, a_{I,\phi} \rangle$.

In general each function $a_{I,\phi}$ may be interpreted as a $G$-equivariant section of the vector bundle
$\scr L_n \to \scr B$. These sections  are pairwise orthogonal by construction. When restricted to the Zariski open subset
$\accentset{\circ}{\scr B}$, they are nowhere trivial, so they define an equivariant basis of $4^n - 5n-1$ sections. 
This implies that the $4^n - 5n-1$ invariants $\theta_{I,\phi}$ are algebraically independent. It is also clear
that adjoining the remaining $2n$ invariants $\lVert a_j\lVert^2, \lVert b_j\rVert^2$ preserves algebraic independence. The result
follows.
\end{proof}

\begin{ex}[$n=2$] In this case $\scr L_2$ is given by the $L$-states $\rho = (a_1,a_2,C)$ in which $a_i \in \scr V_i$ and $C\in \hom(\scr V_1,\scr V_2)$.
Set
\[
b_1 := [C^ta_2,a_1]\, , \quad c_1 := [a_1,b_1]\, , \quad b_2 := [Ca_1,a_2]\, ,\quad
c_2 := [a_2,b_2]\, .
\]
Define $\breve{\scr L}_2 \subset \scr L_2$ be the Zariski open subset consisting of those $\rho = (a_1,a_2,C)$ such that
\begin{itemize}
\item $\lVert a_1\rVert^2 \ne 0 \ne \lVert a_2\lVert^2$;
\item $\lVert b_1\rVert^2 \ne 0 \ne \lVert b_2\lVert^2$.
\end{itemize}
If $\rho \in \breve{\scr L}_2$
then $\{b_i,c_i\}$ is an orthogonal basis for  the transversal subspace $\scr V_i^t$
and $a_i$ is a basis for the longitudinal subspace $\scr V_i^\ell$.
A transcendence basis for $\Bbb C(\breve{\scr L}_2)^G =  \Bbb C(\scr L_2)^G$ is defined by
the nine invariants
\begin{itemize}
\item $\lVert a_1\rVert^2$, $\lVert a_2\lVert^2$,  $\lVert b_1\rVert^2$,  $\lVert b_2\lVert^2$,
\item $C_{rs}^{tt}$, for $r,s \in \{1,2\}$, where
\begin{align*}
C_{11}^{tt} = \langle Cb_1, b_2\rangle\, ,  \quad &C_{12}^{tt} = \langle Cb_1, c_2\rangle\, , \\
 C_{21}^{tt} = \langle Cc_1, b_2\rangle\, , \quad   & C_{22}^{tt} = \langle Cc_1, c_2\rangle\, .
\end{align*}
\item $C^{\ell\ell} = \langle Ca_1, a_2\rangle$.
\end{itemize}
\end{ex}

\section{Second proof of Theorem \ref{bigthm:main}} \label{sec:second}

The second proof of Theorem  \ref{bigthm:main} relies on the notion of a rational section
for the action of $G$ on $\scr L_n$.

 \subsection{Rational sections} Let $G$ act on a variety $X$.
  The {\it normalizer} of a subvariety $S \subset X$ is the subgroup $N\subset G$ 
defined by $\{g\in G\mid gS\subseteq S\}$. Then {\it centralizer} is the normal subgroup
$C\subset N$ consisting of those elements of $G$ which fixes $S$ pointwise. The {\it Weyl group}
is the quotient $W = N/C$.

  \begin{defn}[{\cite[\S2.8]{PV}}]
    \label{defn:relativeSection}
    We say a subvariety $S\subseteq X$ is a \textit{rational section} (with respect to $G$) if there exists a
   dense open subset $S_0 \subset S$ such that
    %$X_{0}\subseteq X$     \begin{enumerate}
      \begin{enumerate} 
      \item The Zariski closure of the orbit of $S$ is $X$, i.e., $\overline{GS}=X$.
      \item If $g \in G$ and $g S_0\cap S_0\neq\emptyset$, then $g\in N$.
    \end{enumerate}
  \end{defn}
  It is not hard to see if $S\subseteq X$ is a rational section then the restriction $\mathbb{C}[X]^{G}\to\mathbb{C}[S]^{W}$ is injective.
  The latter induces an field homomorphism $\mathbb{C}(X)^{G}\to\mathbb{C}(S)^{W}$.
 In fact, 
 
  \begin{prop}[{\cite[p.~161]{PV}}] \label{prop:sec-iso}
    Let $S\subseteq X$ be a rational section with respect to $G$.
    Then the restriction homomorphism $k(X)^{G}\to k(S)^{W}$ is an isomorphism.    
  \end{prop}

We now turn to the second proof of Theorem  \ref{bigthm:main}. We first describe
a rational section for the action of $G$ on $\scr L_n$. This will depend on a fixed choice
of nowhere zero vector field $\Gamma$ as well as
a fixed choice of one-dimensional subspaces
\[
L_i \subset \scr V_i\, , \quad  i \in \overline{n}
\]
with the property that $L_i$ is spanned
by a vector having non-zero squared norm. 

Having fixed the lines $L_i$, we employ the following change in notation:

\begin{notation} If $\rho \in \scr L_n$ is an L-state such that $\rho_i \in L_i$, 
and $ij\in \Gamma_1$, we set $C_{ij} := \rho_{ij}$ and
\begin{equation} \label{eqn:abc}
a_i := \rho_i\, , \quad b_i := [a_i,C_{ij}\lrcorner a_j] \, , \quad c_i = [a_i,b_i]\, .
\end{equation}
Then $a_i,b_i,c_i$ 
are mutually orthogonal vectors in $\scr V_i$. In particular, they form an orthogonal basis
when $\|a_i \|^2, \|b_i\|^2 \ne 0$.
\end{notation}

If we set
\[
\scr V_i^\ell := \Span(a_i) \, , \quad \scr V_i^{t_-} = \Span(b_i) \, , \quad  \scr V_i^{t_+} = \Span(c_i)\, , \quad \scr V_i^t = \Span(b_i,c_i) \, ,
\]
then $\scr V_i^\ell, \scr V_i^{t_-} , \scr V_i^{t_+}$ are mutually orthogonal
subspaces of  $\scr V_i$ and 
\[
\scr V_i^{t} =\scr  \scr V_i^{t_-} \oplus \scr V_i^{t_+} \, .
\]
Furthermore, if $\|a_i\|^2, \|b_i\|^2\ne 0$, then $\scr V_i^\ell, \scr V_i^{t_-} , \scr V_i^{t_+}$ 
are one-dimensional and span $\scr V_i$.

\begin{defn} \label{defn:rational-sec} With respect to $\Gamma$ and $L_\bullet = (L_1,\dots L_n)$ as above,
let
\[
S := S(\Gamma,L_\bullet)  \subset \scr L_n
\] 
denote the complex vector space consisting of all L-states $\rho$ such that 
\begin{itemize} 
\item  $\rho_i \in L_i$ for all $i \in \overline{n}$, and
\item $\langle [C_{ij}\lrcorner a_i,a_j],c_j \rangle = 0$,   for $ij \in \Gamma_1$,  i.e., $[C_{ij}\lrcorner a_i,a_j] \in \scr V_j^{t_-}$. 
\end{itemize}
In the above definition, we have used the notation of \eqref{eqn:abc}.
\end{defn}

\begin{prop} \label{prop:rational-S} The vector space $S$ is a rational section for the action of $G$ on $\scr L_n$. 
Furthermore, the normalizer of $S$ is $(\Bbb Z/2)^{\times n}$ and the centralizer is trivial.
\end{prop}

\begin{proof} We first determine the normalizer and centralizer of $S$.
Recall that $G = \prod_i G_i = \prod_i \SO(\scr V_i)$. 
Let $S_0 \subset S$ denote the Zariski open subset consisting of those $\rho\in S$ such that
$\|a_i\|^2, \|b_i\|^2 \ne 0$ for $i \in \overline{n}$,  where we have employed the notation in
 \eqref{eqn:abc}.

If $\rho \in S_0$, then the vectors $a_i,b_i,c_i$ are an orthogonal frame
in $\scr V_i$. Moreover, the action of $G$ on the frame $a_i,b_i,c_i$ is via $G_i$. 
If $g = (g_1,\dots,g_n) \in G$ acts so that $g\cdot \rho \in S$, then
$g_i$ acts by an orthogonal transformation and fixes each of the lines $\scr V_i^\ell, \scr V_i^{t_-} , \scr V_i^{t_+}$.
As these lines span all of $\scr V_i$, we infer that $g_i = \pm I$. Hence the normalizer is $(\Bbb Z/2)^{\times n}$.
It is clear that the only element of the latter group which stabilizes the vectors $a_i$ is the identity element. Hence, the centralizer is trivial.

To see that $S$ is a rational section, let $\rho \in \scr L_n$ be an $L$-state.
 Then there
exists $g_i \in G_i$ such that $g_i \cdot \rho_i \in L_i$.  Hence, we may assume at the outset that
$\rho_i \in L_i$. 

 Suppose $ij\in \Gamma_1$. Then $[C_{ij}\lrcorner a_i,a_j] \in \scr V_j^t$. We may then choose
$h_j \in \Or(\scr V^t_j)  \subset G_j$ so that $h_j\cdot [C_{ij}\lrcorner a_i,a_j] \in \scr V_j^{t_-}$. Setting
$h = (h_1,\dots,h_n)$, it follows that $h \cdot \rho \in S$, so the condition (1) of Definition
\ref{defn:rational-sec} is satisfied.

With $\rho\in S_0$ as above,
suppose that  for some $g\in G$, say $g = (g_1,\dots,g_n)$, there is  $\rho' \in S_0$ satisfying $g\rho = \rho'$.
Then it is straightforward to check that $g_i$ preserves the lines $\scr V_i^\ell, \scr V_i^{t_-} , \scr V_i^{t_+}$ so $g_i = \pm I$. Consequently,
 condition (2)  of Definition
\ref{defn:rational-sec} is satisfied. We conclude that $S$ is a rational section.
\end{proof}

\begin{proof}[Second proof of Theorem \ref{bigthm:main}] We first observe that the vector space
$S$ has dimension $4^n - 3n-1$.
By Propositions \ref{prop:sec-iso} and \ref{prop:rational-S}, it suffices to show that
$\Bbb C(S)^N$ is rational where $N = (\Bbb Z/2)^{\times n}$. The latter immediately 
follows from a result of E.~Fischer 
\cite{Fischer}, \cite[prop.~4.13]{CT-S} which says that for a finite semi-simple abelian group $G$
acting linearly on a finite dimensional complex
vector space
$V$, the field of rational invariants $\Bbb C(V)^G$ is purely rational and has
transcendence degree  $\dim V$.
\end{proof}


%\bibliographystyle{amsalpha}
%\bibliography{entanglement}

\end{document}